\documentclass[11pt]{article}
\usepackage{a4wide}
\usepackage{amsmath}
\usepackage{amsthm}

\newtheorem{theorem}{Theorem}
\newtheorem{lemma}{Lemma}

\theoremstyle{remark}
\newtheorem{remark}{Remark}
\theoremstyle{definition}

\newcommand{\depth}[1]{\mathrm{Depth}(#1)}
\newcommand{\Depth}[2]{\mathrm{Depth}_{#2}(#1)}
\newcommand{\clique}[1]{\mathrm{CLIQUE}(#1)}
\newcommand{\iclique}[1]{\mathrm{CLIQUE}[#1]}
\newcommand{\mcliq}[1]{\kappa(#1)}
\newcommand{\blog}{\log }
\newcommand{\Cliq}[1]{\mathrm{Cliq}(#1)}
\newcommand{\Cl}[1]{\mathrm{Cl}(#1)}
\newcommand{\thr}[2]{\mathrm{Th}_{#1}^{#2}}
\newcommand{\is}[1]{\omega(#1)}
\newcommand{\isb}[1]{\omega_{\mathrm{b}}(#1)}
\newcommand{\cc}[1]{\mathrm{c}(#1)}
\newcommand{\ccb}[1]{\mathrm{c}_{\mathrm{b}}(#1)}
\newcommand{\Circ}{C}

\begin{document}

\title{Clique Problem, Cutting Plane Proofs and
       Communication Complexity}

\author{S.~Jukna\thanks{Research supported by the DFG
    grant SCHN~503/5-1.}\thanks{University of Frankfurt,
    Institute of Computer Science,
    D-60054 Frankfurt, Germany. Permanent affiliation: Vilnius University,
    Institute of Mathematics and informatics,
    Vilnius, Lithuania. Email: jukna@thi.informatik.uni-frankfurt.de}}
\date{}
\maketitle

  \begin{abstract}
    Motivated by its relation to the length of cutting plane proofs
    for the Maximum Biclique problem, we consider the following
    communication game on a given graph $G$, known to both
    players. Let $K$ be the maximal number of vertices in a complete
    bipartite subgraph of $G$, which is not necessarily an
    induced subgraph if $G$ is not bipartite.
    Alice gets a set $a$ of vertices, and
    Bob gets  a disjoint set $b$ of vertices such that $|a|+|b|>K$.
    The goal is to find a nonedge of $G$ between $a$ and $b$.  We show
    that $O(\log n)$ bits of
    communication are enough for every $n$-vertex graph.
  \end{abstract}

 \section{Introduction}

Let $G=(V,E)$ be a graph with vertex set $V$ and edge set $E$.  A
\emph{clique} in $G$ is a set $a\subseteq V$ of vertices such that
$\{u,v\}\in E$ for all $u\neq v\in a$.  A \emph{biclique} in $G$ is a
pair $\{a,b\}$ of disjoint subsets of vertices such that $\{u,v\}\in
E$ for all $u\in a$ and $v\in b$. Thus, the edges $\{u,v\}$ form a
complete bipartite subgraph of $G$ (which is not necessarily an
\emph{induced} subgraph if $G$ is not bipartite).  The \emph{size} of
a clique (or biclique) is the number of its vertices.  The maximum
size of a clique in $G$ is denoted by $\is{G}$, and the maximum size
of a biclique in $G$ is denoted by $\isb{G}$.  Note that $\is{G}\leq
\isb{G}$ holds for every graph $G$: every clique of size $k$ contains
a biclique (in fact, many bicliques) of size~$k$.  A \emph{nonedge} in
a graph is a pair of its nonadjacent vertices.

Given an arbitrary (not necessarily bipartite) graph $G=(V,E)$, we are
interested in the communication complexity of the following game
between two players, Alice and Bob.
\begin{itemize}
\item[] {\bf Biclique Game on} $G=(V,E)$:\\
  Alice gets $a\subseteq V$, Bob gets $b\subseteq V$ such that $a\cap
  b=\emptyset$ and $|a|+|b|>\isb{G}$.  The goal is to find a nonedge
  of $G$ lying between $a$ and $b$. This nonedge must be known to both
  players.
\end{itemize}
If the underlying graph $G$ is bipartite with a bipartition $V=V_1\cup
V_2$, then we additionally require that $a\subseteq V_1$ and
$b\subseteq V_2$.  Note that the promise $|a|+|b|>\isb{G}$ ensures
that there must be at least one nonedge between $a$ and $b$. The communication complexity,
$\ccb{G}$, of this game is the minimum, over all (deterministic)
communication protocols for $G$, of the number of bits communicated on
a worst-case input $(a,b)$. We stress that the graph $G$ in this game
is \emph{fixed} and is known to both players. The players are not
adversaries---they help and trust each other. The difficulty, however,
is that Alice cannot see Bob's set $b$, and Bob cannot see Alice's
set~$a$.

To avoid trivialities, we will assume (without mentioning this) that
our graphs have no complete stars, that is, vertices adjacent to all
remaining vertices---such vertices can be ignored.

\begin{itemize}
\item[] {\bf Clique Game on} $G=(V,E)$:\\
  Alice gets a set $a\subseteq V$ on vertices, Bob gets a set
  $b\subseteq V$ of vertices such that $a\cap b=\emptyset$ and
  $|a|+|b|>\is{G}$.  The goal is to find a nonedge of $G$ lying within
  $a\cup b$. Again, this nonedge must be known to both players.
\end{itemize}
Let $\cc{G}$ denote the communication complexity of the clique game
on~$G$.
\begin{remark}\label{rem:1}
  The main difference from the biclique game is that now we have a
  weaker promise $|a|+|b|>\is{G}$. Note also that the only nontrivial
  inputs are pairs $(a,b)$, where both $a$ and $b$ are cliques: the
  found nonedge must then lie \emph{between} $a$ and $b$ (as in the
  biclique game). Indeed, if one of the sets, say, $a$ is not a
  clique, then it contains a nonedge. Alice can then send both
  endpoints of this nonedge to Bob using at most $2\lceil \log_2
  n\rceil$ bits, and the game is over.
\end{remark}

Our motivation to consider clique and biclique games cames 
 from their connection to the length of so-called ``tree like''
cutting plane proofs for the Maximum Clique problem on a fixed graph
$G=(V,E)$. Cliques in $G$ are exactly the $0$-$1$ solutions of the
system $\Cl{G}$ consisting of linear inequalities $x_u+x_v\leq 1$ for
all nonedges $\{u,v\}\not\in E$, and $x_v\geq 0$ for all vertices
$v\in V$. If the graph is bipartite with bipartition $V=V_1\cup V_2$,
then we only have inequalities $x_u+x_v\leq 1$ for all nonedges $\{u,v\}$
with $u\in V_1$ and $v\in V_2$. 
In the ``find a hurt axiom'' game, given a $0$-$1$
assignment $\alpha$ to the variables such that $\sum_{v\in
  V}\alpha_v\geq \is{G}+1$, we (the adversary) first split the bits of
$\alpha$ between Alice and Bob, and their goal is to find a nonedge
$\{u,v\}$ such that $\alpha_u=\alpha_v=1$. In the bipartite case,
the promise is $\sum_{v\in V}\alpha_v\geq \isb{G}+1$.

Results of \cite{impagliazzo} imply that, if a clique or biclique
game requires $K$ bits of communication, then every tree-like cutting
planes proof of the $0$-$1$ unsatisfiability of the system $\Cl{G}$
augmented by the inequality $\sum_{v\in V}x_v\geq \is{G}+1$ must
either use super-polynomially large coefficients, or must produce at
least $2^{\Omega(K/\log n)}$ inequalities; see \cite[Section~19.3
and Research Problem~19.12]{BOOK} for details.  It was therefore a
hope that $n$-vertex graphs $G$ with $\ccb{G}\gg \log^2 n$ or at least
$\cc{G}\gg \log^2 n$ exist.

Our main result (Theorem~\ref{thm:1} below) destroys the first hope:
for every (not necessarily bipartite) $n$-vertex graph,
$\ccb{G}=O(\log n)$ bits of communication are enough.

Since the found nonedge must be known to both players, at least $\blog
n$ bits of communication are necessary for any non-trivial graph on
$n$-vertices.  However, if the graph is complicated enough, then
(intuitively) this trivial number of bits should not be sufficient.
If, say, there are many nonedges leaving the sets $a$ and $b$, but
only one of them lies between $a$ and $b$, how should the players
quickly localize this unique nonedge? 

It turns out that, somewhat surprisingly, a logarithmic number of bits
is sufficient for \emph{any} graph! That is, up to constant factors,
the communication complexity of the biclique game does not depend on
the structure of the underlying graph.

\begin{theorem}\label{thm:1}
  For every $n$-vertex graph $G$, we have $\ccb{G}\leq 7.3\blog
  n+O(1)$.
\end{theorem}

The situation with the clique game is more complicated. Here we are
only able to show that $O(\log n)$ bits are enough for many graphs.
Interestingly, the clique game is related to the monotone complexity
of the following decision problem.

The \emph{induced $k$-clique function} of an $n$-vertex graph $G$ is a
monotone boolean function of $n$ variables which, given a subset of
vertices, outputs $1$ if and only if some $k$ of these vertices form a
clique in~$G$.  Thus, this function is just a version of the
well-known {\bf NP}-complete Clique function restricted to only
\emph{spanning} subgraphs of one \emph{fixed} graph $G$.  Let
$\depth{G}$ denote the maximum, over all integers $1\leq k\leq n$, of
the minimum depth of a monotone circuit with fanin-$2$ AND and OR
gates computing the induced $k$-clique function of~$G$.

\begin{theorem}\label{thm:2}
  For every $n$-vertex graph $G$, we have $\cc{G}\leq \depth{G}+2\blog
  n+O(1)$.
\end{theorem}

The measure $\depth{G}$ is related to the number $\mcliq{G}$ of
\emph{maximal} cliques in $G$; a clique is \emph{maximal}, if it
cannot be extended by adding a new vertex.  It can be shown (see
Lemma~\ref{lem:6} below) that
\begin{equation}\label{eq:cliq}
  \depth{G}\leq \blog\mcliq{G}+5.3\blog n+O(1)\,.
\end{equation}
There are many $n$-vertex graphs $G=(V,E)$ for which $\mcliq{G}$ is
polynomial in $n$. In particular, $\mcliq{G}\leq n(d/2)^{p-2}$ holds
for every $K_p$-free graph of maximal degree $d\geq 2$ \cite{prisner};
$\mcliq{G}\leq n^p$, where $p$ is the chromatic number of $G$
\cite{moon}; $\mcliq{G}\leq (|E|/p+1)^p+|E|$, where $p$ is the maximum
number of edges in an induced matching in the complement of~$G$
\cite{balas,alekseev}.  If $p=O(\log n)$ then Theorem~\ref{thm:2}
gives $\cc{G}=O(\log^2 n)$ for all such graphs, implying that
communication complexity arguments will fail for such graphs, even for
the Maximum Clique problem (not just for the Maximum Biclique
problem).

Still, it remains unknown whether $\cc{G}=O(\log^2 n)$ holds for all
graphs.  We can only show that $O(\log n)$ bits are always enough in
the following relaxed version of the clique game. This version is no
more related to cutting plane proofs, but may be of independent
interest.

A \emph{common neighbor} of a subset $b\subseteq V$ of vertices is a
vertex $v\not\in b$ which is adjacent to all vertices in $b$. Let
$\Gamma(b)$ denote the set of all such common neighbors.

\begin{itemize}
\item[] {\bf Relaxed Clique Game on} $G=(V,E)$:\\
  Alice gets a set $a\subseteq V$ on vertices, Bob gets a set
  $b\subseteq V$ of vertices such that $a\cap b=\emptyset$ and
  $|a|+|b|>\is{G}$.  The goal is to find a nonedge of $G$ which lies
  either within $a\cup b$ or between $a$ and $\Gamma(b)$.
\end{itemize}

That is, the found nonedge must lie either within $a\cup b$ (as in the
clique game) or between $a$ and~$\Gamma(b)$.

\begin{theorem}\label{thm:3}
  In the relaxed clique game, $7.3\blog n+O(1)$ bits of communication
  are enough for every $n$-vertex graph.
\end{theorem}

\section{The biclique game: proof of Theorem~\ref{thm:1}}

Let $G=(V,F)$ be a graph on $|V|=n$ vertices with edge set $F$.
Inputs to the biclique game on $G$ are pairs $(a,b)$ of disjoint
subsets of vertices such that $|a|+|b|>\isb{G}$.  Hence, there must be
at least one nonedge lying between $a$ and $b$.  The goal is to find
such a ``crossing'' nonedge.

To solve this task, let $E:=\tbinom{V}{2}\setminus F$ be the set of
all nonedges of~$G$, and take a set $X=\{x_e\colon e\in E\}$ of
boolean variables, one for each nonedge. Say that a nonedge $e$ is
\emph{incident} with a subset $a\subseteq V$, if $e\cap
a\neq\emptyset$. For a subset $a\subseteq V$ of vertices, let
$E(a)\subseteq E$ denote the set of all nonedges incident with~$a$.
Finally, we associate with every subset $a\subseteq V$ two vectors
$p_a$ and $q_a$ in $\{0,1\}^{|E|}$ whose coordinates correspond to
nonedges $e\in E$:
\begin{itemize}
\item $p_a(e)=1$ if and only if $e\in E(a)$;
\item $q_a(e)=0$ if and only if $e\in E(a)$.
\end{itemize}
Thus, $p_a$ is the characteristic vector of $E(a)$, and 
$q_a$ is the complement of $p_a$.  Given an input $(a,b)$,
the goal in the biclique game is to find a position (a nonedge) $e$
such that $p_a(e)=1$ ($e$ is incident with $a$) and $q_b(e)=0$ ($e$ is
incident with $b$).  To do this, we will use monotone circuits for
threshold functions.  Recall that a \emph{threshold}-$k$ function
$\thr{k}{n}$ accepts a $0$-$1$ vector of length $n$ if and only if it
contains at least $k$ ones.  By a \emph{monotone circuit} we will mean
a circuit consisting of fanin-$2$ AND and OR gates; no negated
variables are allowed as inputs. The \emph{depth} of a circuit is the
length of a longest path from an input to the output gate.

\begin{theorem}[Valiant \cite{valiant}]\label{thm:val}
  Every threshold function $\thr{k}{n}$ can be computed by a monotone
  circuit of depth at most $5.3\log n+O(1)$.
\end{theorem}
We will use this result to show that there exist at most $n$
small-depth monotone circuits such that every given pair of vectors
$(p_a,q_b)$ is separated by at least one of them. Then we use these
circuits to design the desired protocol.

\begin{lemma}\label{lem:1}
  For every $1\leq k\leq n$, there is a monotone circuit $\Circ(X)$ of
  depth at most $6.3\blog n+O(1)$ such that $\Circ(p_a)=1$ and
  $\Circ(q_b)=0$ for all subsets $a$ and $b$ of vertices of size
  $|a|=k$ and $|b| > \isb{G}-k$.
\end{lemma}

\begin{proof}
  Associate with each subset $c\subseteq V$ the monomial
  \[
  M_c(X):=\bigwedge_{e\in E(c)}x_e\,,
  \]
  and let $f_k(X)$ be the OR of these monomials over all $k$-element
  subsets $c\subseteq V$. Then $f_k$ clearly accepts vector $p_a$ for
  every $k$-element subset of vertices $a$. So, let $b\subseteq V$ be
  a subset of $|b| > \isb{G}-k$ vertices.  To show that the function
  $f_k$ rejects the vector $q_b$, it is enough to show that every its
  monomial $M_c$ does this.

  \emph{Case 1}: $c\cap b=\emptyset$.  Since $|c|=k$ and $c\cap
  b=\emptyset$, our assumption $|c|+|b|>|c|+(\isb{G}-k)=\isb{G}$
  implies that there must be a nonedge between $c$ and $b$, that is, a
  nonedge $e$ in $E(c)\cap E(b)$.  But vector $q_b$ sets all variables
  $x_e$ with $e\in E(b)$ to $0$, implying that $M_c(q_b)=0$.

  \emph{Case 2}: $c\cap b\neq \emptyset$. Since we assumed that $G$
  contains no complete stars, there must be a nonedge $e$ incident to
  some vertex in $a\cap b$. So, $e\in E(c)\cap E(b)$, and we again
  obtain that $M_c(q_b)=0$.

  Thus, $f_k(p_a)=1$ and $f_k(q_b)=0$ for all disjoint subsets $a$ and
  $b$ of vertices of size $|a|=k$ and $|b| > \isb{G}-k$. It therefore
  remains to show that the function $f_k$ can be computed by a
  monotone circuit $\Circ$ of depth at most $6.3\blog n+O(1)$.

  The function $f_k$ accepts a set $E'\subseteq E$ of nonedges if and
  only if $E(c)\subseteq E'$ holds for some subset $c\subseteq V$ of
  $|c|=k$ nodes, which happens if and only if $E'$ contains at least
  $k$ of the sets $E(v)=\{e\in E\colon v\in e\}$ of nonedges incident
  to vertices $v$.  We can therefore construct a monotone circuit
  $\Circ(X)$ computing $f_{k}(X)$ as follows.

  The circuit, testing whether $E(v)\subseteq E'$, is just the AND
  $M_v(X)=\bigwedge_{e\in E(v)}x_e$ of at most $n$ variables. Thus, by
  taking the threshold-$k$ of the outputs of these ANDs, we obtain an
  unbounded fanin circuit of depth-$2$ computing $f_k$. Each $M_v$ has
  a monotone fanin-$2$ circuit of depth at most $\blog n+1$. By
  Theorem~\ref{thm:val}, the function $\thr{k}{n}$ has such a circuit
  of depth at most $5.3\blog n+O(1)$. Thus the depth of the entire
  circuit is at most $6.3\blog n+O(1)$.
\end{proof}

We can now describe our protocol for the biclique game on the graph
$G=(V,F)$. Recall that inputs to this game are pairs $(a,b)$ of
disjoint subsets of vertices such that $|a|+|b|>\isb{G}$.

Alice first uses at most $\blog n+1$ bits to communicate Bob the size
$k=|a|\leq \isb{G}$ of her set $a$; hence $|b|>\isb{G}-k$.  The
players then take a minimal monotone circuit $\Circ$ guaranteed by
Lemma~\ref{lem:1}. Hence, $\Circ(p_a)=1$ and $\Circ(q_b)=0$.  After
that they traverse (as in \cite{KW}) the circuit $\Circ$ backwards
starting at the output gate by keeping the invariant: $\Circ'(p_a)=1$
and $\Circ'(q_b)=0$ for every reached subcircuit $\Circ'$.

Namely, suppose the output gate of $\Circ$ is an AND gate, that is, we
can write $\Circ=\Circ_0 \wedge \Circ_1$.  Then Bob sends a bit $i$
corresponding to a function $\Circ_i$ such that $\Circ_i(q_b)=0$; if
both $\Circ_0(q_b)$ and $\Circ_1(q_b)$ output $0$, then Bob sends $0$.
Since $\Circ(p_a)=1$, we know that $\Circ_i(p_a)=1$.  If
$\Circ=\Circ_0 \lor \Circ_1$, then it is Alice who sends a bit $i$
corresponding to a function $\Circ_i$ such that $\Circ_i(p_a)=1$;
again, if both $\Circ_0(p_a)$ and $\Circ_1(p_a)$ output $1$, then
Alice sends $0$.  Since $\Circ(q_b)=0$, we know that $\Circ_i(q_b)=0$.

Alice and Bob repeat this process until they reach an input of the
circuit.  Since the circuit is monotone (there are no negated inputs),
this input is some variable $x_e$.  Hence, $x_e(p_a)=1$ and
$x_e(q_b)=0$. By the definition of vectors $p_a$ and $q_b$ (and since
$a\cap b=\emptyset$), this means that the nonedge $e$ lies between $a$
and $b$, as desired.

The number of communicated bits in this last step is at most the depth
$6.3\blog n+O(1)$ of the circuit $\Circ$.  Thus, the total number of
communicated bits is at most $7.3\blog n+O(1)$. This completes the
proof of Theorem~\ref{thm:1}.  \qed

\begin{remark}
  One could presume that the main reason, why the biclique game has
  small communication complexity, is just the fact that the biclique
  problem \emph{is} solvable in polynomial time via, say, the maximum
  matching algorithm. In the biclique problem, we are given a graph
  $G$ and a positive integer $K$; the goal is to decide whether $G$
  contains a biclique $a\times b$ of size $|a|+|b|\geq K$. However, it
  is known \cite{peeters} that a similar \emph{maximum edge biclique}
  problem is already {\bf NP}-complete, even for bipartite graphs.  In
  this problem, the goal is to decide whether $G$ contains a biclique
  $a\times b$ with $|a\times b|\geq K$ edges. If $G$ is a graph, in
  which every biclique has at most $K$ edges, then the corresponding
  to this latter problem game is, given two disjoint sets $a,b$ of
  vertices such that $|a\times b|>K$, to find a nonedge between $a$
  and $b$. It is easy to see that $O(\log n)$ bits of communication
  are enough also in this game. For this, it is enough just to replace
  the condition $|b|>\isb{G}-k$ in Lemma~\ref{lem:1} by the condition
  $|b|>K/k$. The rest of the proof is the same.
\end{remark}

\section{The clique game: proof of Theorem~\ref{thm:2}}

Consider the clique game for a given $n$-vertex graph
$G=(V,F)$. Inputs to this game are pairs $(a,b)$ of disjoint subsets
of vertices such that $|a|+|b|>\is{G}$, and the goal is to find a
nonedge lying within $a\cup b$.  Hence, now the promise is weaker, but
also the task is (apparently) easier: it is allowed that the found nonedge lies
within $a$ or within~$b$.

Let us first see why we cannot use the same function $f_k$ as in the
biclique game.  Recall that $f_k$ is the OR of monomials
$M_c(X)=\bigwedge_{e\in E(c)}x_e$ over \emph{all} $k$-element subsets
$c\subseteq V$.  Now, even if $b\subseteq V\setminus c$ is a clique,
the condition $|c|+|b|>\is{G}$ does not imply that $M_c(q_b)=0$. If,
for example, there are no nonedges lying between $c$ and $b$, that is,
when all nonedges in $c\cup b$ lie within the set $c$, then $q_b(e)=1$
for all nonedges $e\in E(c)$, implying that $M_c(q_b)=1$, that is, the
function $f_k$ wrongly accepts the vector~$q_b$. To get rid of this
problem, we use more complicated circuits.

\begin{lemma}\label{lem:2}
  For every $1\leq k\leq n$, there is a monotone circuit $\Circ(X)$ of
  depth at most $\depth{G}+\blog n$ such that $\Circ(p_a)=1$ and
  $\Circ(q_b)=0$ for all cliques $a$ and $b$ of size $|a|=k$ and $|b|
  > \is{G}-k$.
\end{lemma}

\begin{proof}
  As before, associate with each subset $c\subseteq V$ the monomial
  $M_c(X):=\bigwedge_{e\in E(c)}x_e$, and let $g_k(X)$ be the OR of
  such monomials over all $k$-cliques $c\subseteq V$.  That is, we now
  take the OR only over sets $c$ containing no nonedges. Let
  $b\subseteq V$ be a clique of size $|b| > \is{G}-k$.  If $c\cap
  b\neq\emptyset$, then the star-freeness of $G$ implies $E(c)\cap
  E(b)\neq \emptyset$, and hence, also $M_c(q_b)=0$.  If $c\cap
  b=\emptyset$, then $|c|+|b|>\is{G}$ implies that there must be a
  nonedge in $c\cup b$. But since both $c$ and $b$ are cliques, this
  nonedge must lie between $c$ and $b$, that is belong to $E(c)\cap
  E(b)$, and we again obtain that $M_c(q_b)=0$. Thus, $g_k(p_a)=1$ and
  $g_k(q_b)=0$ for all cliques $a$ and $b$ of size $|a|=k$ and $|b| >
  \is{G}-k$.

  To design a monotone circuit of desired depth for the function
  $g_k$, recall that $g_k$ accepts a set $E'\subseteq E$ of nonedges
  if and only if there is a $k$-clique $c\subseteq V$ such that
  $M_v(E')=1$ for all $v\in c$.  Thus, applying the induced $k$-clique
  function of $G$ to the outputs of the monomials $M_v$, we obtain a
  monotone circuit for $g_k$ of depth at most $\depth{G}$.
\end{proof}

We can now describe our protocol for the clique game on a given graph
$G=(V,F)$. By Remark~\ref{rem:1}, we can assume that the inputs are
pairs $(a,b)$ of disjoint cliques such that $|a|+|b|>\is{G}$.  The
goal is to find a nonedge lying between $a$ and~$b$.

Using at most $\blog n+1$ bits, Alice first communicates Bob the size
$k=|a|\leq \is{G}$ of her clique $a$; hence $|b|>\is{G}-k$.  The
players then take a minimal monotone circuit $\Circ$ guaranteed by
Lemma~\ref{lem:2}.  Hence, $\Circ(p_a)=1$ and $\Circ(q_b)=0$. By
traversing this circuit, the players will find a variable $x_e$ (an
input of $\Circ$) such that $x_e(p_a)=1$ and $x_e(q_b)=0$. By the
definition of vectors $p_a$ and $q_b$ (and since $a\cap b=\emptyset$),
this means that the nonedge $e$ lies between $a$ and $b$, as desired.
\qed

We now prove the inequality \eqref{eq:cliq}.  Note that
Theorem~\ref{thm:val} states that $\depth{K_n}\leq 5.3\blog n+O(1)$.
The graph $K_n$ has only one maximal clique---the graph itself.  But
Valiant's theorem can be easily extended to graphs with a larger
number of maximal cliques. Recall that $\mcliq{G}$ denotes the number
of maximal cliques in~$G$.

\begin{lemma}\label{lem:6}
  For every $n$-vertex graph $G$, $\depth{G}\leq
  \blog\mcliq{G}+5.3\blog n+O(1)$.
\end{lemma}

\begin{proof}
  Let $G=([n],E)$ be a graph, and $\Cliq{x}$ be its induced $k$-clique
  function.  That is, $\Cliq{x}=1$ if and only if the set
  $S_x=\{i\colon x_i=1\}$ contains a $k$-clique of~$G$. Since every
  clique is contained in some maximal clique, we have that
  $\Cliq{x}=1$ if and only if $|S_x\cap C|\geq k$ for at least one
  maximal clique $C$. Thus, if $\delta_C\in\{0,1\}^n$ is the
  characteristic vector of $C$, and if $\delta_C\land x$ is a
  component-wise AND, then $\Cliq{x}=1$ if and only if
  $\thr{k}{n}(\delta_C\land x)=1$ holds for at least one maximal
  clique $C$. By taking the OR, over all $\mcliq{G}$ maximal cliques
  $C$, of monotone circuits computing the threshold functions
  $\thr{k}{n}(\delta_C\land x)$, and using Theorem~\ref{thm:val}, we
  obtain a monotone circuit of depth at most $\blog\mcliq{G}+5.3\blog
  n+O(1)$ computing~$\Cliq{x}$.
\end{proof}

\section{Relaxed clique game: proof of Theorem~\ref{thm:3}}

Let $G=(V,F)$ be a graph on $|V|=n$ vertices.  Inputs to the relaxed
clique game on $G$ are pairs $(a,b)$ of disjoint subsets of vertices
with the same promise $|a|+|b|>\is{G}$ as in the clique game. The
task, however, is easier: the found nonedge must either lie within
$a\cup b$ (as in the clique game) or between $a$ and some common
neighbor of~$b$.  We will argue as before, but will use a modified
definition of Bob's vectors~$q_b$.

Namely, say that a nonedge is a \emph{common neighbor} of set
$b\subseteq V$, if both its endpoints are common neighbors of $b$,
that is, are connected (by edges of $G$) to all vertices in~$b$.  Now
define the vector $q_b'$ by: $q_b'(e)=0$ if and only if $e\in E(b)$ or
$e$ is a common neighbor of~$b$.

\begin{lemma}\label{lem:3}
  For every $1\leq k\leq n$, there is a monotone circuit $\Circ(X)$ of
  depth at most $6.3\blog n+O(1)$ such that $\Circ(p_a)=1$ and
  $\Circ(q_b')=0$ for all cliques $a$ and $b$ of size $|a|=k$ and $|b|
  > \is{G}-k$.
\end{lemma}

\begin{proof}
  Let $f_k(X)$ be the monotone boolean function defined in the proof
  of Lemma~\ref{lem:1}.  That is, $f_k$ is the OR of monomials
  $M_c(X)=\bigwedge_{e\in E(c)}x_e$ over all $k$-element sets
  $c\subseteq V$. Let $b\subseteq V$ be a clique of size $|b| >
  \is{G}-k$.  It is enough to show that every monomial $M_c$ rejects
  the vector~$q_b'$.  This clearly holds if $E(c)\cap
  E(b)\neq\emptyset$, because $q_b'$ sets to $0$ all variables $x_e$
  with $e\in E(b)$.

  So, assume that $E(c)\cap E(b)=\emptyset$, that is, $c\cap
  b=\emptyset$ and there are no nonedges between $c$ and $b$. Since
  $b$ is a clique, the condition $|c|+|b|>|c|+(\is{G}-k)=\is{G}$
  implies that both endpoints of some nonedge $e$ must belong to $c$.
  But the absence of nonedges between $c$ and $b$ implies that $e$ is
  common neighbor of $b$. Hence, again, the vector $q_b'$ sets the
  variable $x_e$ to $0$, and $M_c(q_b')=0$ holds.

  Thus, $f_k(p_a)=1$ and $f_k(q_b')=0$ for all disjoint cliques $a$
  and $b$ of size $|a|=k$ and $|b| > \is{G}-k$. Since, as shown in the
  proof of Lemma~\ref{lem:1}, the function $f_k$ can be computed by a
  monotone circuits of depth at most $6.3\blog n+O(1)$, we are done.
\end{proof}

The protocol for the relaxed clique game on a graph $G$ is now the
same as for the clique game. As in that game, interesting are only
inputs $(a,b)$, where both $a$ and $b$ are cliques. In this case, the
players take the circuit guaranteed by Lemma~\ref{lem:3}, and traverse
it until they find a nonedge $e=\{u,v\}$ such that $x_e(p_a)=1$ and
$x_e(q_b')=0$. By the definition of vectors $p_a$ and $q_b'$, this
means that one endpoint of $e$, say, vertex $u$ belongs to the clique
$a$, and the second endpoint $v$ either belongs to $b$ or is a common
neighbor of $b$ (because in this latter case the nonedge $e$ must be a
common neighbor of $b$).  In both cases, the nonedge $e$ is a legal
answer in the relaxed clique game.  \qed

\begin{remark}
  Note that if $a$ and $b$ are disjoint cliques such that
  $|a|+|b|>\is{G}$, then there \emph{must} be a ``crossing'' nonedge
  (between $a$ and $b$), which would be a legal answer in the clique
  game. However, the protocol for the relaxed game may output a
  ``wrong'' nonedge---a common neighbor of~$b$.
\end{remark}

\section{Conclusion and open problems}

Note that our communication protocol is not explicit because the
construction of a small-depth monotone circuits for the majority
function in \cite{valiant} is probabilistic.  To get an explicit
protocol, one can use the construction of a circuit of depth $K\blog
n$ for the majority function given in \cite{ajtai}.  But the constant
$K$ resulting from this construction is huge, it is about~$5000$.

The main message of Theorem~\ref{thm:1} is that communication
complexity arguments cannot yield any non-trivial lower bounds on the
length of cutting plane proofs for systems corresponding to the
Maximum Biclique problem, because $\ccb{G}=O(\log n)$ holds for all
$n$-vertex graphs $G$. However, the case of the Maximum Clique problem
remains unclear.  Do $n$-vertex graphs $G$ requiring $\cc{G}\gg \log^2
n$ bits of communication in the clique game exist?  We have only shown
that $\cc{G}=O(\log n)$ holds for a lot of graphs, and that this
number of communicated bits is enough for all graphs in the relaxed
clique game (which is no more related to cutting plane proofs).

Let us mention that a different type of (adversarial) games,
introduced in \cite{PI}, was recently used in \cite{lauria} to derive
strong lower bounds for tree like \emph{resolution} proofs for the
Maximum Clique problem.  Is there some analogue of these games in the
case of cutting plane proofs?

The clique and biclique games on a given graph $G$ are special cases
of a \emph{monotone} Karchmer--Wigderson game \cite{KW}: given a pair
$(A,B)$ of two intersecting subsets of a fixed $n$-element set, find
an element in their intersection $A\cap B$. (In our case we have
$A=E(a)$ and $B=E(b)$.)  In the \emph{non-monotone} game, inputs are
pairs of distinct sets, and the goal is to find an element in the
symmetric difference $A\oplus B:=(A\setminus B)\cup(B\setminus A)$.
It is usually much easier to find an element in the symmetric
difference than in the intersection.  Say, if the players know that
$|A|\neq |B|$, $O(\log n)$ bits are also enough to find an element in
$A\oplus B$~\cite{brodal}.  However, monotone games (with the goal to
find an element in the intersection) usually require much more bits of
communication. For example, the \emph{monotone} game corresponding to
the matching problem requires $\Omega(n)$ bits of communication
\cite{RW}, whereas \cite{brodal} implies that $O(\log n)$ bits are
enough in the \emph{non-monotone} game for this problem.  It is
therefore interesting that, in the biclique game, a logarithmic number
of communicated bits is enough even to find an element in the
intersection~$A\cap B$, not just in~$A\oplus B$.

Finally, it would be interesting to understand the (monotone)
complexity of the induced $k$-clique functions $\iclique{G,k}$, that
is, to prove nontrivial lower bounds on $\Depth{G}{k}$, the smallest
depth of a monotone circuit computing this function for individual
graphs~$G$.  Recall that $\iclique{G,k}$ accepts a set of vertices if
and only if the induced subgraph of $G$ on these vertices contains a
$k$-clique.

The minterms of $\iclique{G,k}$ are $k$-cliques of $G$, and maxterms
are $k$-\emph{clique transversals}, that is, minimal sets of vertices
intersecting all $k$-cliques of $G$. Thus, the result of Karchmer and
Wigderson \cite{KW} implies that $\Depth{G}{k}$ is exactly the
communication complexity of the following game for $\iclique{G,k}$:
Alice gets a $k$-clique, Bob a $k$-clique transversal, and the goal is to find a
common vertex. Theorem~\ref{thm:2} shows that the communication
complexity $\cc{G}$ of the clique game is at most $\Depth{G}{k}$ plus
an additive logarithmic factor.  Does some reasonable converse (up to
an additive $\log ^2 n$ factor) of this inequality hold?  What is
$\Depth{G}{k}$ for random graphs~$G$?

In the communication game for the {\bf NP}-complete problem
$\clique{n,k}$, inputs are pairs $(A,B)$ of subsets of edges (not
vertices) of $K_n$ such that edges in $A$ form a $k$-clique, and edges
in $B$ form a $k$-\emph{coclique}, that is, $B$ consists of $k-1$
vertex-disjoint cliques covering all vertices of $K_n$. The goal is to
find an edge in $A\cap B$. It is known that, for particular choices of
$k=k(n)$, this game requires $\Omega(\sqrt{k}\log n)$ bits
\cite{razborov,AB}, and even $\Omega(n^{1/3})$ bits \cite{GH} of
communication. Can the arguments of \cite{razborov,AB,GH} be adopted
to the game for $\iclique{G,k}$?  The problem in this latter game is
with Bob's inputs: how to find a large family of $k$-clique transversals in $G$
such that only a small fraction of them will contain a fixed set of,
say, $\sqrt{n}$ vertices?  Actually, it is even not clear whether
there exist a sequence $(G_n: n=1,2,\ldots)$ of $n$-vertex graphs
$G_n$ for which $\iclique{G_n,k}$ is an {\bf NP}-complete problem.

\section*{Acknowledgments}

I am thankful to Mario Szegedy for interesting initial discussions on
the biclique game, and to Jacobo Tor\'an for detecting an error in a
previous protocol for the clique game.

\end{document}